\documentclass{article}
\usepackage{amsmath}
\usepackage{amsthm}
\usepackage{amsfonts}
\usepackage{graphicx}
\usepackage[usenames,dvipsnames]{color}
\usepackage{authblk}

\renewcommand{\int}[1]{$(\mathit{int}\cup{#1})$}
\newcounter{thecase} 
	
\newcommand{\case}[1]{\noindent
            \refstepcounter{thecase}\paragraph*{\textnormal{\textit{Case \arabic{thecase}: #1}}}
}

\newcommand{\bb}[1]{\ensuremath{\mathbf{#1}}}

\newtheorem{theorem}{Theorem}[section]
\newtheorem{lemma}[theorem]{Lemma}
\newtheorem{definition}[theorem]{Definition}

\newtheorem{observation}[theorem]{Observation}

\begin{document}

%\begin{frontmatter}

\title{1-String CZ-Representation of Planar Graphs}

\author{T.~Biedl\thanks{\texttt{biedl@uwaterloo.ca}, Research supported by NSERC.}~}
\author{M.~Derka\thanks{\texttt{mderka@uwaterloo.ca}, The second author was supported by the NSERC Vanier CGS.}}

\affil{\small{David R.~Cheriton School of Computer Science, University of Waterloo}}

\renewcommand\Authands{~and~}
%\author[uwaterloo]{T.~Biedl\fnref{nserc}}
%\ead{biedl@uwaterloo.ca}
%\fntext[nserc]{\add{Acknowledge your NSERC here.}}

%\author[uwaterloo]{M.~Derka\fnref{vanier}}
%\ead{mderka@uwaterloo.ca}
%\fntext[vanier]{The work of the second author was supported by the NSERC Vanier CGS.}

%\address[uwaterloo]{David R.~Cheriton School of Computer Science, University of Waterloo, Waterloo, ON, Canada}
%
\maketitle

\begin{abstract}
In this paper, we prove that every planar 4-connected graph has a CZ-representation---a
string representation using paths in a rectangular grid that contain at most one vertical segment. 
Furthermore, two paths representing vertices $u,v$ intersect precisely once whenever there
is an edge between $u$ and $v$. The required size of the grid is $n \times 2n$.
\end{abstract}

%\begin{keyword}
%Planar graphs\sep String representations\sep $B_k$-VPG graphs\sep Graph drawing\sep Intersection graphs
%\end{keyword}
%
%\end{frontmatter}

\section{Preliminaries}

A possible way of representing graphs is to assign to every vertex a curve so
that two curves cross if and only if there is an edge between the respective vertices. 
Here, two curves $\bb{u},\bb{v}$ \emph{cross} means that
they share a point $s$ internal to both of them and 
the boundary of a sufficiently small closed disk around $s$ 
is crossed by $\bb{u},\bb{v},\bb{u},\bb{v}$ (in this order).
The representation of graphs using crossing curves is referred to as a \emph{string representation},
and graphs that can be represented in this way are called \emph{string graphs}.

In 1976, Ehrlich, Even and Tarjan showed that every planar graph has a string representation~\cite{cit:tarjan}.
It is only natural to ask if this result holds if one is restricted to using
only some ``nice'' types of curves. In 1984, Scheinerman conjectured that all planar graphs can
be represented as intersection graphs of line segments~\cite{cit:scheinerman}.
This was proved first for bipartite graphs~\cite{cit:arroyo, cit:pach} with the strengthening
that every segment is vertical or horizontal. The result was extended to triangle-free
graphs, which can be represented by line segments with at most three distinct slopes~\cite{cit:castro}.
%The class of graphs which can be represented using segments with at most $k$ slopes are
%called {$k$-DIR}. 

Since Scheinerman's conjecture seemed difficult to prove for all planar
graphs, interest arose in possible relaxations.
Note that any two line segments intersect at most once.
Define 1-STRING to be the class of graphs that are intersection graphs
of curves (of arbitrary shape) that intersect at most once.
The original construction of string representation for planar graphs 
given in~\cite{cit:tarjan} requires curves to cross multiple times. 
In 2007, Chalopin, Gon\c{c}alves and Ochem showed that every
planar graph is in 1-STRING~\cite{cit:chalopin-gonclaves-ochem, cit:chalopin-string}.  With respect to Scheinerman's
conjecture, while the argument of~\cite{cit:chalopin-gonclaves-ochem, cit:chalopin-string} shows that the prescribed number
of intersections can be achieved, it provides no idea on the complexity of curves that is required.   

Another way of restricting curves in string representations is to require them
to be \emph{orthogonal}, i.e., to be paths in a grid.  Call a graph a
{\em VPG-graph} (as in ``Vertex-intersection graph of Paths in a Grid'')
if it has a string representation with orthogonal curves.
It is easy to see that all planar graphs are VPG-graphs (e.g.~by generalizing
the construction of Ehrlich, Even and Tarjan).  For bipartite planar graphs,
curves can even be required to have no bends \cite{cit:arroyo, cit:pach}.
For arbitrary planar graphs bends are required in orthogonal curves, and
recently Chaplick and Ueckerdt showed that 2 bends per curve always suffice
\cite{cit:chaplick}.  Let {\em $B_2$-VPG} be the graphs that have
a string representation where curves are
orthogonal and have at most 2 bends; the result in
\cite{cit:chaplick} then states that planar graphs are in $B_2$-VPG.
Unfortunately, in Chaplick and Ueckerdt's construction, curves may cross 
each other repeatedly, and so it
does not prove that planar graphs are in 1-STRING.

The conjecture of Scheinerman remained open until 2009 when it was proved true by Chalopin and Gon\c{c}alves \cite{cit:chalopin-gonclaves-ochem}
who extended the technique used to prove their 1-STRING 
result~\cite{cit:chalopin-seg}. 

\paragraph*{Our results:} 
In this paper, we show that every planar 4-connected graph has a 
string representation that simultaneously satisfies the requirements for
1-STRING (any two curves cross at most once) and the requirements
for $B_2$-VPG (any curve is orthogonal and has at most two bends).
Our result hence re-proves, in one construction, the results
by Chalopin et al.~\cite{cit:chalopin-seg} and the result by
Chaplick and Ueckerdt \cite{cit:chaplick}, albeit only for 4-connected 
planar graphs. (We briefly discuss extensions in Section~\ref{sec:outlook}.)

In our construction all curves have one of four
possible shapes: C-shape, Z-shape, or their mirror images. We call such a 
representation a \emph{CZ-representation} (see Section~\ref{sec:def} for
the formal definition).

\begin{theorem}
\label{thm:main-claim}
Every $4$-connected planar graph has a 1-string CZ-rep\-re\-sen\-ta\-tion. 
\end{theorem}
%
% TB: removed the corollary, because this is really nothing new: Chalopin
% and Chaplick put together proves it, too
%\begin{corollary}
%The class of $4$-connected planar graphs belongs to 1-STRING $\cap$ $B_2$-VPG.
%\end{corollary}

Since our construction has at most $n$ vertical and $2n$ horizontal line segments, and
since any orthogonal grid can be deformed to be on integer coordinates 
without empty rows
or columns, our construction can be embedded to a rectangular grid of size $n \times 2n$.
Note that none of the previous results provided an intuition of the required size of the grid.

Our approach is inspired by the construction of 1-string representations from 
2007~\cite{cit:chalopin-gonclaves-ochem, cit:chalopin-string}. 
The authors proved the result in two steps. First,
they showed that triangulations without separating triangles 
admit 1-string representations. By induction on the number of 
separating triangles, they then showed that 1-string representation
exists for any planar triangulation, and consequently for any 
planar graph. 

In order to show that triangulations without separating triangles
have 1-string representation, Chalopin et at.~\cite{cit:chalopin-string} used
a method inspired by Whitney's proof that 4-connected planar graphs
are Hamiltonian~\cite{cit:whitney}. Asano, Saito and Kikuchi later improved
Whitney's technique and simplified his proof~\cite{cit:ham-cycle}. 
Our paper uses the same approach as~\cite{cit:chalopin-string}, but borrows ideas from~\cite{cit:ham-cycle}
and develops them further to reduce the number of cases and hence
simplify the proof.

\iffalse
A~CZ-representation uses two bends for every curve, thus our result
implies that triangulations with no separating triangles as 2-VPG. 
Our construction uses precisely one intersection per edge. Thus, it is 
stronger than the result of~\cite{cit:chaplick} for this specific subclass
of planar graphs.

Furthermore, the result presented here immediately implies that planar triangulations
with no separating triangles have 1-string representations, which is an intermediate
step of~\cite{cit:chalopin-string} for showing that this is true for all planar graphs.
Compared to~\cite{cit:chalopin-string}, our approach significantly reduces the number of analysed cases.
\fi

\section{CZ-Representation of $\mathbf{4}$-Connected Planar Graphs}
\label{sec:def}

Let us begin with a formal definition of a \emph{CZ-representation}.

\begin{definition}[CZ-representation]
\label{def:repre}
 A planar graph $G$ has a \emph{1-string CZ-representation} if every vertex $v$ of $G$ can be
 represented by a curve $\mathbf{v}$ such that: \begin{enumerate}
    \item Curve $\mathbf{v}$ is \emph{orthogonal}, i.e., it consists of horizontal and vertical segments.
    \item Curve $\mathbf{v}$ has at most two bends and at most one vertical segment (see Figure~\ref{fig:cz-curve}).
    \item Curves $\mathbf{u}$ and $\mathbf{v}$ intersect at most once, and $\bb{u}$ intersects $\bb{v}$ if and
	only if $(u,v)$ is an edge of $G$.
 \end{enumerate}
A {\em partial 1-string CZ-representation} is a 1-string CZ-representation of a subgraph of $G$.
 \end{definition}

\begin{figure}
\centering
    \includegraphics[width=.7\textwidth]{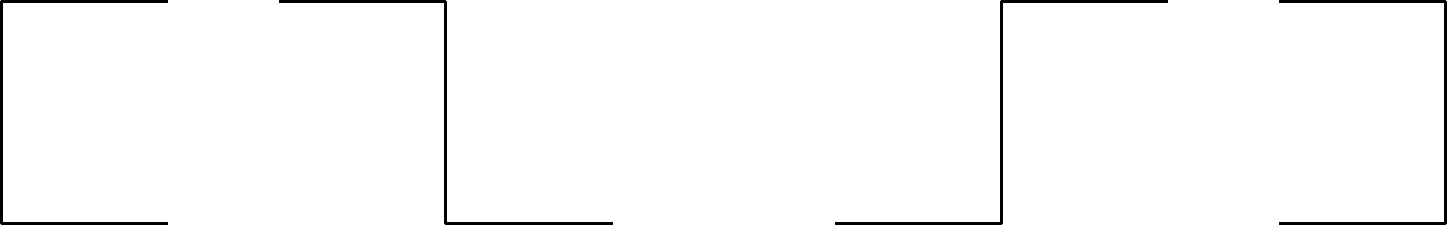}
    \caption{Every curve in a CZ-representation has one of the depicted shapes.}
    \label{fig:cz-curve}
\end{figure}

For brevity, we use ``CZ-representation'' to mean ``partial 1-string 
CZ-repre\-sentation''.
Our technique for constructing a CZ-representation of a graph uses an intermediate step
referred to as ``an \emph{\int{F}-CZ-representation}\footnote{Here, \emph{int} is used as an abbreviation of \emph{interior}.} of a
W-triangulation that satisfies the chord condition with respect to
three chosen corners''.
We define these terms first. 

A \emph{triangulated disk} is a 2-connected planar graph $G$ such that every 
interior face is a triangle. 
A {\em separating triangle} is a cycle of length $3$ which contains vertices
both inside and outside.   Following the notation of
\cite{cit:chalopin-string}, a \emph{W-triangulation} is a triangulated disk
which does not contain a separating triangle.
A {\em chord} of a triangulated disk is an
interior edge for which both endpoints are on the outer face.

For two vertices $X, Y$ on the outer face of a connected planar graph,
define $P_{XY}$ to be the counter-clockwise (ccw) path on the outer face from $X$ to $Y$ ($X$ and $Y$ inclusive).
We will often study triangulated disks with
three specified distinct vertices $A,B,C$ called the {\em corners}
which must appear on the outer face in ccw order. 
We denote $P_{AB} = (a_1, a_2, \ldots, a_r)$, $P_{BC} = (b_1, b_2, \ldots, b_s)$ 
and $P_{CA} = (c_1,c_2,\ldots,c_t)$, where $c_t = a_1 = A$, $a_r = b_1 = B$ 
and $b_s = c_1 = C$. 

\begin{definition}[Chord condition]
\label{def:chord-condition}
A W-triangulation $G$ satisfies the \emph{chord condition} with respect
to the corners $A,B,C$ if $G$ has no chord within $P_{AB}, P_{BC}$ or $P_{CA}$,
i.e., no interior edge of $G$ has 
both ends on $P_{AB}$, or both ends of $P_{BC}$, or
both ends on $P_{CA}$.%
\footnote{For readers familiar with \cite{cit:chalopin-string}
or \cite{cit:ham-cycle}:
A W-triangulation that satisfies the chord condition with respect
to corners $A,B,C$ is called a \emph{W-triangulation 
with 3-boundary $P_{AB},P_{BC},P_{CA}$} 
in \cite{cit:chalopin-string},
and the chord condition is 
the same as \emph{Condition (W2b)} in~\cite{cit:ham-cycle}.}%
\end{definition}

\begin{definition}[\int{F}-CZ-representation]
Let $G$ be a connected planar graph with corners $A,B,C$. 
Let $F$ be a set of outer face edges incident to $C$.
An \emph{\int{F}-CZ-representation}%
\footnote{An \int{F}-CZ-representation corresponds roughly to what Chalopin
et al.~\cite{cit:chalopin-string} call Property 1, except that they do not
restrict the shape of the curves, and they fix $F$ to be edge $(C,c_2)$.}
of $G$ is a CZ-representation of $G$ for which curves $\bb{u},\bb{v}$ cross if and
only if $(u,v)$ is an interior edge of $G$ or $(u,v) \in F$. 
Furthermore, the CZ-representation must satisfy that:
\begin{enumerate}
    \item There exists a rectangle $\Theta$ containing all intersections 
    of curves so that the top of $\Theta$ is intersected, from right to 
	left in order, by the curves of the vertices of $P_{AB}$, 
%    by  $\bb{A} = \bb{a_1}, \bb{a_2}, \ldots, \bb{a_r} = \bb{B}$ 
   and the bottom
    of $\Theta$ is intersected, from left to right in order, 
%by $\bb{B} = \bb{b_1}, \bb{b_2}, \ldots, 
%    \bb{b_s} = \bb{C} = \bb{c_1}, \bb{c_2}, \ldots, \bb{c_t} = \bb{A}$.
	by the curves of the vertices of $P_{BA}$. 
    \item The curve $\bb{v}$ of an outer face vertex $v$ has at most one
    bend. (By (1), this implies that $\bb{A}$ and $\bb{B}$ have no bends.)
\end{enumerate}
\end{definition}

See Figure~\ref{fig:partial-cz-ex} for examples of an \int{F}-CZ-representation.
In all our constructions, we have $|F|\leq 1$, i.e., $F$ consists of
at most one edge that is on the outer face and incident to $C$.
If $F=\{e\}$, then $e$ is called the \emph{special edge}.
We sometimes write \int{e}-CZ-representation rather than 
\int{\{e\}}-CZ-representation, and
int-CZ-representation rather than
\int{\emptyset}-CZ-representation.
Note that the roles of corners $A$ and $B$ in an \int{F}-CZ-representation 
are symmetric: we can exchange $A$ and $B$, as long as we also reverse
all cyclic orders of edges around each vertex (to preserve the sense
of counter-clockwise) and flip the resulting representation horizontally
(to undo the reversal.)
Corner $C$, on the other hand, is distinct from the other two, for 
example because the special edge must be incident to $C$. 
%, since $\bb{C}$ intersects
%$\Theta$ only once and because $\bb{C}$ is allowed to have one bend. 
Our key result is the following:

\begin{figure}
    \centering
    \includegraphics[width=\textwidth]{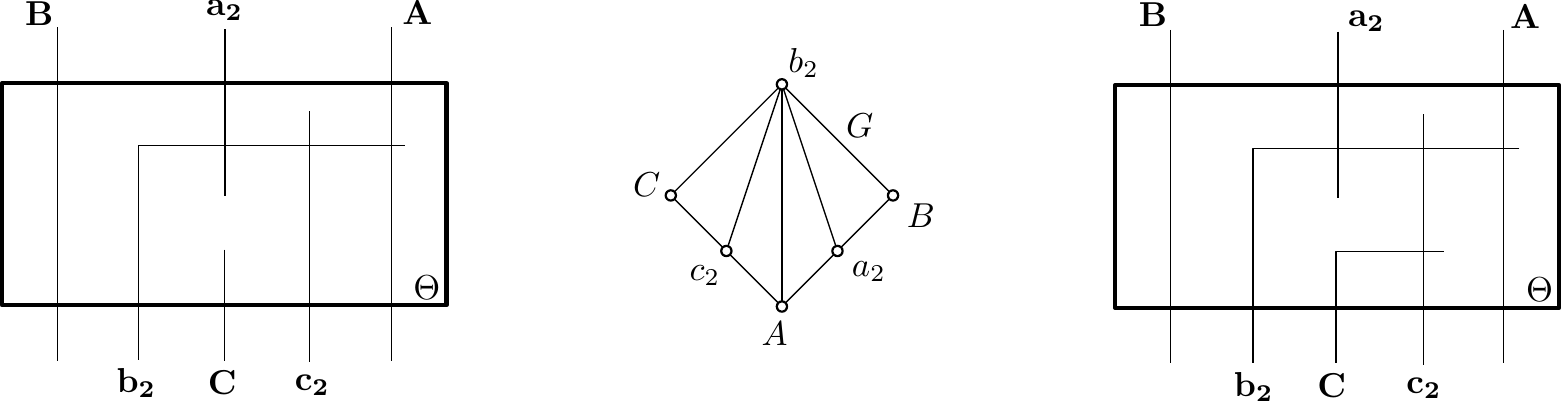}
    \caption{An int-CZ-representation (left) 
and an \int{(C,c_2)}-CZ-representation (right).}
    \label{fig:partial-cz-ex}
\end{figure}

\begin{lemma}
	\label{lem:representation}
    Let $G$ be a W-triangulation that 
    satisfies the chord condition with respect to corners $A,B,C$. 
	Then $G$ has an \int{F}-CZ-representation
    for any set $F$ of at most one outer face edge incident to $C$.
\end{lemma}

The proof of Lemma~\ref{lem:representation} will be given in Section~\ref{sec:proof}. Here we show how it implies our main result.

\begin{proof}[Proof of Theorem~\ref{thm:main-claim}]
First assume that $G$ is a triangulation, which by 4-con\-nec\-tivity means that
it has no separating triangles. Let $A,B,C$ be the vertices on the outer face
in ccw order.
As the outer face is a triangle, $G$ clearly satisfies the chord condition with respect to $A,B,C$.
Thus, by Lemma~\ref{lem:representation}, it has an
\int{(B,C)}-CZ-representation contained in a rectangular box $\Theta$.  This
CZ-representation has an intersection for 
every edge except for $(A,B)$ and $(A,C)$. The  ends of curves $\bb{A}$ and $\bb{B}$
outside of $\Theta$ can be used to create intersections for these edges
as follows.
%In order to represent edge $(A,B)$,
Bend and stretch the upper end of $\bb{B}$ rightwards and the upper end of $\bb{A}$ upwards so that
both the curves cross. 
%Repeat the same construction below $\Theta$ to represent edge $(A,C)$. 
Bend and stretch the lower end of curve $\bb{A}$ leftwards and stretch $\bb{C}$ downwards so that the two curves cross.
Recall that $\bb{A}$ and $\bb{B}$ initially did not have any bends, so they have each one bend in the
constructed 1-string 
$CZ$-representation of $G$. See Figure~\ref{fig:completion} for an illustration.

Now assume that $G$ is a $4$-connected planar graph. Then \emph{stellate} the graph, i.e.,
insert a vertex into each non-triangulated face and connect it to all vertices on that face.
The resulting graph is triangulated and has no separating triangles, so it has a 1-string CZ-representation
by the above. Deleting the curves of added vertices produces the result. 
\end{proof}

\begin{figure}
    \centering
    \includegraphics[width=.25\textwidth]{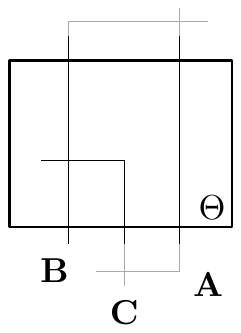}
    \caption{Completing the \int{F}-CZ-representation of a triangulation $G = (A,B,C;B)$.}
    \label{fig:completion}
\end{figure}

\section{\int{F}-CZ-representations}
\label{sec:proof}

In this section, we provide the proof of Lemma~\ref{lem:representation}.
We proceed by induction on the number of edges. In the base case, $n=3$, so
$G$ is a triangle, and the three corners $A,B,C$ must be the three vertices 
of this triangle.  The \int{F}-CZ-representations 
for $F=\emptyset, \{(A,C)\}, \{(B,C)\}$
are depicted in Figure~\ref{fig:base-case}.

\begin{figure}
	\centering
	\includegraphics[width=.8\textwidth]{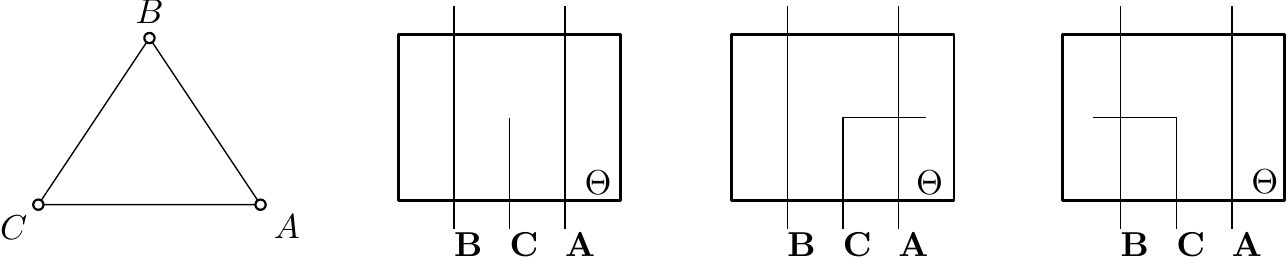}
	\caption{\int{F}-CZ-representations of a triangle.}
	\label{fig:base-case}
\end{figure}

The induction step for $n \geq 4$ is divided into five cases.

\case{$G$ has a chord incident to $C$.}  
\label{case:special}
By the chord condition, this chord has the form $(C,a_i)$ for some $1 < i < r$.
The graph $G$ can be split along the chord $(C,a_i)$ into two graphs $G_1$ and
$G_2$.  Both $G_1$ and $G_2$ are bounded by simple cycles, hence triangulated
disks.  No edges were added, so neither $G_1$ nor $G_2$
contains a separating triangle.   We select $(A,a_i,C)$ as corners for $G_1$
and $(B,C,a_i)$ as corners for $G_2$ and can easily verify that with this
$G_1$ and $G_2$ satisfy the chord condition: 
\begin{itemize}
\item $G_1$ has no chords on $P_{Aa_i}$ or $P_{CA}$ as they would violate the chord condition in $G$. 
There is no chord on $P_{a_iC}$ as it is a single edge.
\item $G_2$ has no chords on $P_{a_iB}$ or $P_{BC}$ as they would violate the chord condition in $G$. 
There is no chord on $P_{a_iC}$ as it is a single edge.
\end{itemize}

So, by induction, Lemma~\ref{lem:representation} holds for both $G_1$ and $G_2$.

\begin{figure}
\centering
\includegraphics[width=.35\textwidth]{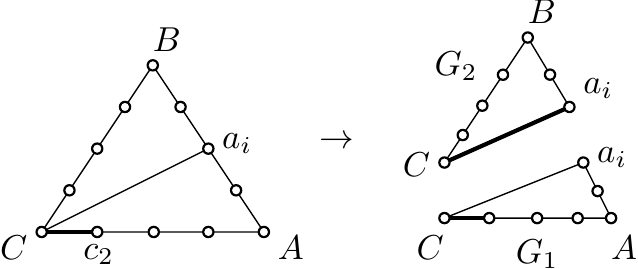}\hspace{4em}
\includegraphics[width=.2\textwidth]{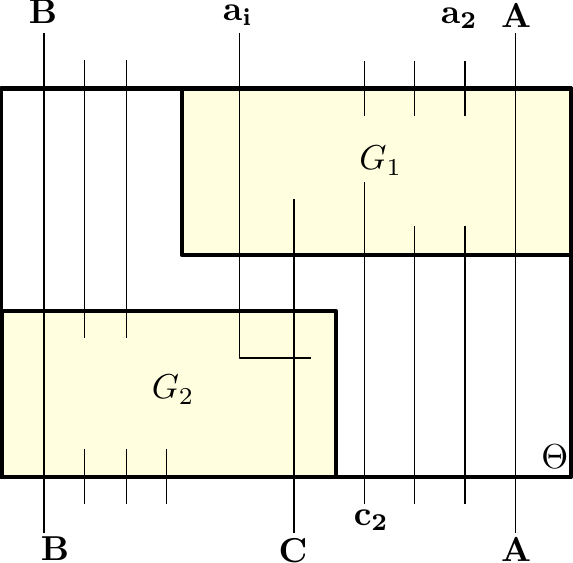}\hspace{4em}
\includegraphics[width=.2\textwidth]{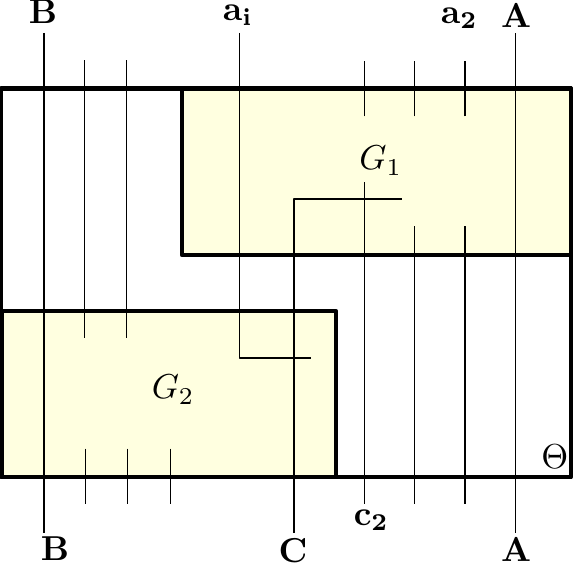}
\caption{Case~\ref{case:special}: Constructing an 
\int{F}-CZ-representation for $F=\emptyset$ (left) and $F=\{(C,c_2)\}$ 
(right) when $C$ is incident with a chord.}
\label{fig:case0}
\end{figure}

For $F=\emptyset$ or $F=\{(C,c_2)\}$,
an \int{F}-CZ-representation of $G$ 
can be constructed as follows. Inductively, construct an 
\int{F}-CZ-representation of $G_1$ and an 
\int{(C,a_i)}-CZ-representation of $G_2$; note that the special edge
of each indeed attaches at the required corner. Rotate
the CZ-representation of $G_2$ by 180$^\circ$, and translate it
so that it is below the CZ-representation of $G_1$ with the two
copies of $\bb{a_i}$ in the same column.  Stretch one of the CZ-representations 
horizontally as needed until the two copies of $\bb{c_j}$ are also in
the same column; then $\bb{a_i}$ and $\bb{c_j}$ can each be unified without
adding bends by adding vertical segments.    Stretch the CZ-representation
of $G_2$ further so that everything to the left of $\bb{a_i}$ in $G_2$
appears to the left of the entire CZ-representation of $G_1$; the curves
of outer face vertices of $G$ then cross (after suitable lengthening)
a bounding box in the required order.  
See also Figure~\ref{fig:case0}.
%Observe that no other curve crosses the lower boundary of the rectangle bounding the
%CZ-representation of $G_2$ and the upper boundary of the rectangle bounding the CZ-representation
%of $G_1$ inbetween $\bb{a_i}$ and $\bb{C}$. 
The construction does not create any new bends. Since $\mathbf{a_i}$ has no bends in the CZ-representation of $G_1$ and $\bb{C}$ has
no bends in the CZ-representation of $G_2$, the number of bends on each curve on the outer face does not exceed $1$.

This finishes the construction for $F=\emptyset$ or $F=\{(C,c_2)\}$.  
An \int{(C,b_{s-1})}-CZ-representation can be obtained using the
same construction after exchanging the roles of vertices $A$ and $B$
as described earlier.

\case{$G$ has a chord in the form $(a_i,c_j)$, $1 \leq i \leq r$, $1 \leq j \leq t$.}
\label{case:case1}
We may assume that $j > 1$ (otherwise we are in Case~\ref{case:special}), 
and $i>1$ and $j<t$ (otherwise the chord condition is violated).

\begin{figure}
\centering
\includegraphics[width=.5\textwidth]{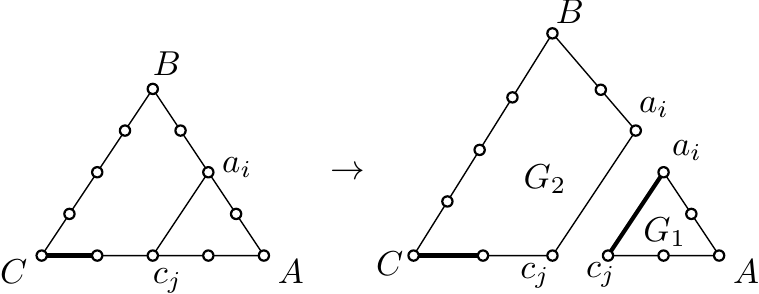}\hspace{5em}
\includegraphics[width=.3\textwidth]{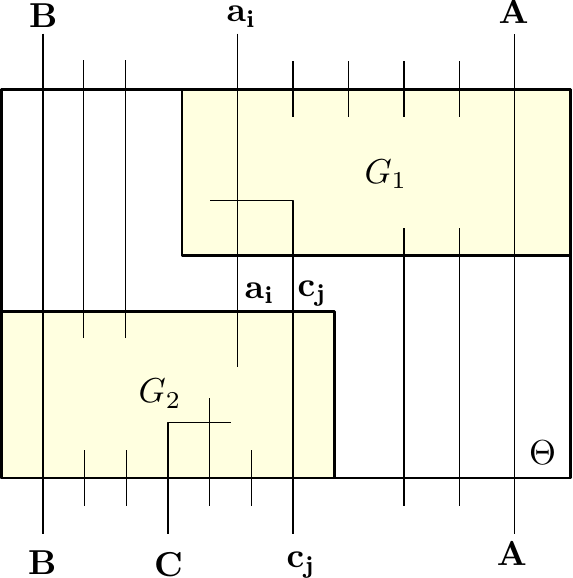}
\caption{Case~\ref{case:case1}: Constructing an int-CZ-representation of $G=(A,B,C)$ with a chord of 
the form $(a_i,c_j)$, $1 < i \leq r$, $1 < j < t$.}
\label{fig:case1-detail}
\end{figure}

Let $(a_i,c_j)$ be the chord that maximizes $i-j$ (i.e., it is the furthest chord from vertex $A$). 
Note that possibly $i = r$, i.e., $a_i = a_r = B$. 
Split the graph $G$ along this chord into graphs  $G_1$ (which contains $A$)
and $G_2$ (which contains $C$ and the special edge, if any).
Select $(A,a_i,c_j)$ as corners for $G_1$ and $(c_j,B,C)$ as corners for $G_2$.
As before both $G_1$ and $G_2$ are W-triangulations, and we can verify
that they satisfy the chord condition:
\begin{itemize}
\item $G_1$ has no chords on path $P_{Aa_i} \subseteq P_{AB}$ and $P_{c_jA} \subseteq P_{CA}$ as they would contradict the chord condition
in $G$. The remaining side $P_{c_ja_i}$ is a single edge, and so does not have any chords either.
\item $G_2$ has no chords on path $P_{Cc_j} \subseteq P_{CA}$ and $P_{BC}$ as they would contradict
the chord condition for $G$. Furthermore, $G_2$ has no chord on path $P_{c_jB}$ due to the selection
of $(a_i,c_j)$ and since $P_{AB}$ has no chords.
\end{itemize}

%Both $G_1$ and $G_2$ have strictly fewer edges than $G$. 
In order to construct an \int{F}-CZ-representation of $G$,
apply induction to get an \int{(a_i,c_j)}-CZ-representations of 
$G_1$ and an \int{F}-CZ-representation of 
$G_2$. Similarly as before, stretch the representations so that we can 
align and join the two 
curves $\mathbf{a_i}$ and $\mathbf{c_j}$ as shown in Figure~\ref{fig:case1-detail}. Since $\mathbf{a_i}$ 
has no bends in the CZ-representation of $G_1$ and $\mathbf{c_j}$ has no bends in the CZ-representation of $G_2$, the number of bends of those curves in the constructed representation is at most~1. The
number of bends of any other curve representing an outer face vertex does not change, so it is also at most $1$. Thus, the constructed representation is a valid
\int{F}-CZ-representation.
Figure~\ref{fig:case1-detail} shows the construction 
(with $F=\emptyset$) when $a_i \neq B$ and $2 < j $. 
Figure~\ref{fig:case1-special} shows the construction 
when $a_i = B$ or $c_j = c_2$.

\begin{figure}
\begin{center}
\includegraphics[width=.2\textwidth]{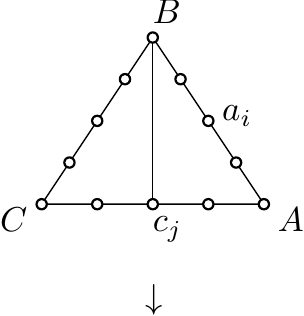}\hspace{4em}
\includegraphics[width=.2\textwidth]{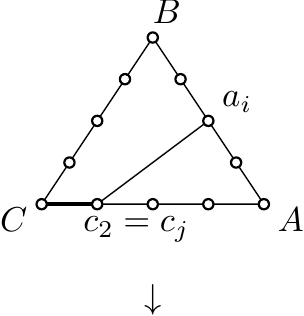}\hspace{4em}
\includegraphics[width=.2\textwidth]{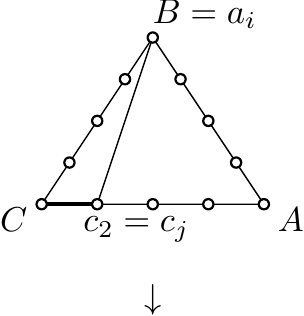}\\\medskip
\includegraphics[width=.2\textwidth]{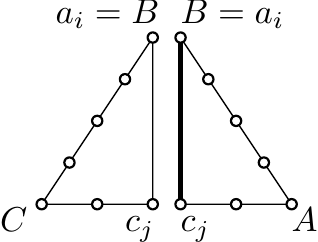}\hspace{4em}
\includegraphics[width=.2\textwidth]{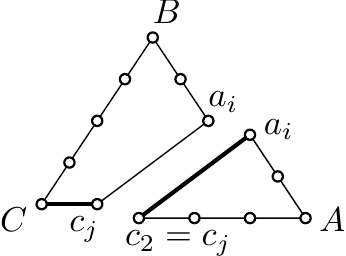}\hspace{4em}
\includegraphics[width=.2\textwidth]{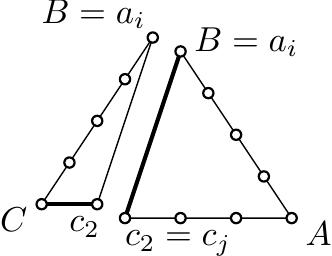}\\\medskip\medskip
\includegraphics[width=.2\textwidth]{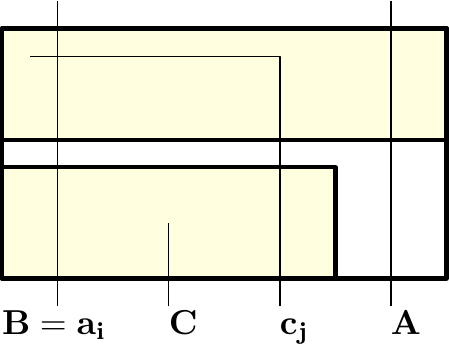}\hspace{4em}
\includegraphics[width=.2\textwidth]{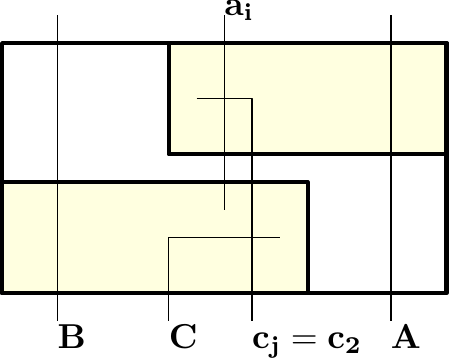}\hspace{4em}
\includegraphics[width=.2\textwidth]{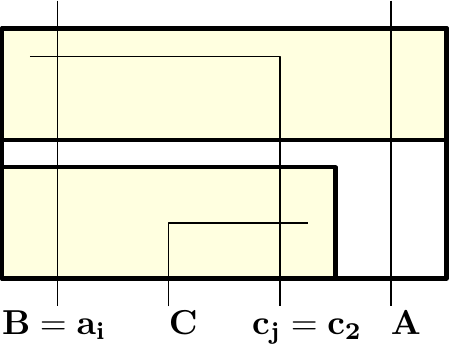}
\end{center}
\caption{The construction for Case~\ref{case:case1} also covers borderline cases:  (Left) The chord $(a_i,c_j)$ is incident with $B$. (Middle) The chord
is incident to the end $c_2$ of the special edge.  (Right) Both the incidencies are present.}
\label{fig:case1-special}
\end{figure}

\case{$G$ has a chord in the form $(a_i, b_j)$, $1 \leq i \leq r$, $1 \leq j \leq s$.}
\label{case:case1a}
By interchanging the roles of corners $A$ and $B$, this case can be transformed into Case~\ref{case:case1}.

\case{$G$ has a chord in the form $(b_j, c_k)$, $1 \leq j \leq r$, $1 \leq k \leq t$.}
\label{case:case2}
Note that we may assume $1 < j < r$ and $1 < k < t$ as all other cases either violate the chord condition
or were already covered.
Let $(b_j,c_k)$ be the chord maximizing $j-k$ (i.e., furthest from $C$). 

In order to construct an int-CZ-representation of $G$, split the graph 
along $(b_j,c_k)$ into two W-triangulations $G_1$ (which includes $C$
and the special edge, if any) and
$G_2$ (which includes $A$).  Set
 $(A,B,c_k)$ as corners for $G_1$ and $(c_k,b_j,C)$ as corners for $G_2$
and verify the chord condition:

\begin{itemize}
    \item $G_1$ has no chords on either $P_{Cc_k} \subseteq P_{CA}$ or $P_{b_jC} \subseteq P_{BC}$
        as they would contradict the chord condition in $G$.
        The third side is a single edge $(b_j,c_k)$ and so it does not have any chords either. 
    \item $G_2$ has no chords on either $P_{c_kA} \subseteq P_{CA}$ or $P_{AB}$ as they
    would violate the chord condition in $G$. It does not have any chords on the path $P_{Bc_k}$ due to 
    the selection of the chord $(b_j,c_k)$. 
\end{itemize}

Thus, by induction, $G_1$ has an \int{F}-CZ-representation and $G_2$ has
an \int{(b_j,c_k)}-CZ-representation. 
After horizontal deformation, the CZ-representations can be aligned 
so that the ends of $\bb{b_j}$ and $\bb{c_k}$ in $G_2$ can be connected to the upper ends of $\bb{b_j}$
and $\bb{c_k}$ in $G_1$. 
As $\bb{b_j}$ and $\bb{c_k}$ have no bends in the CZ-representation of $G_2$, the
construction does not increase the number of bends on any curve and produces 
an \int{F}-CZ-representation of $G$. 
Figure~\ref{fig:case2-construction} shows the construction.

\begin{figure}
\begin{center}
\includegraphics[width=.25\textwidth]{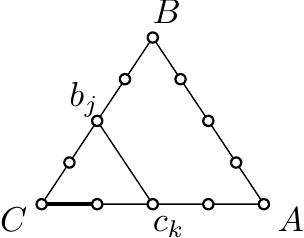}\hspace{2em}
\includegraphics[width=.25\textwidth]{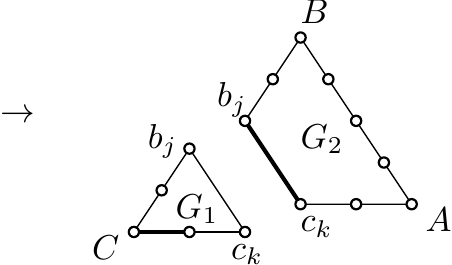}\hspace{3em}
\includegraphics[width=.25\textwidth]{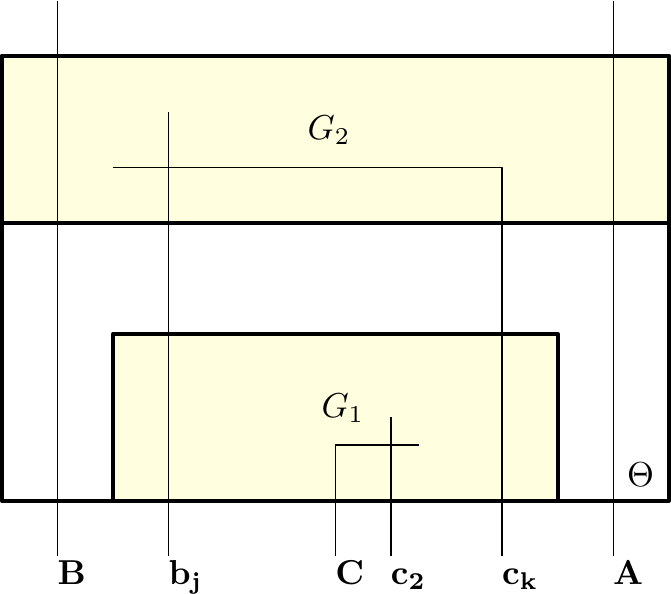}
\end{center}
\caption{Case~\ref{case:case2}: Construction of an \int{F}-CZ-representation of $G = (A,B,C)$
    with a chord $(b_j,c_k), 1 \leq j \leq r, 1 \leq k \leq t$.}
\label{fig:case2-construction}
\end{figure}

\case{$G$ has no chords.}
\label{case:chordless}
Assume after possible exchange of $A$ and $B$ that the special edge, if it exists, is $(C,c_2)$.
Let $u_1,\dots,u_q$ be the neighbours of vertex $C$ in clockwise order, starting
with $b_{s-1}$ and ending with $c_2$.  We know that $\deg(C)\geq 2$, for
otherwise the neighbours of $C$ would have a chord between them since $G$ is a 
triangulated disk.  Therefore, $C$ has at least one neighbour
$u_i, 1 < i < q$.  We also know that $u_2,\dots,u_{q-1}$ are not on the outer 
face, since $C$ is not incident to a chord.

Let $u_j$ be a neighbour of $C$ that has at least one other neighbour on 
$P_{CA}$, and among all those, choose $j$ to be minimal. Such a
$j$ exists and $j<q$  because $G$ is triangulated and 
therefore $u_{q-1}$ is adjacent to both $C$ and $u_q$.   
We also know that $j>1$, sicne otherwise there would be a chord from
$u_j=u_1=b_{s-1}$ to some vertex on $P_{CA}$.

Let the {\em terminals} be the neighbours of $u_j$
on $P_{CA}$; we denote these by $t_1,t_2,\dots,t_x$ in the order in which
they appear on $P_{CA}$. 
Separate $G$ into two graphs $G_T$ (the \emph{top} graph) and $G_B$ (the \emph{bottom} graph) as follows: $G_T$ is bounded 
by $(u_1, u_2, \ldots, u_j, t_x, \stackrel{P_{t_xA}}{\ldots}, A, \stackrel{P_{AB}}{\ldots}, B, \stackrel{P_{Bu_1}}{\ldots}, u_1)$;
and $G_B$ is bounded by $(C = t_1, \stackrel{P_{t_1t_x}}{\ldots}, t_x, u_j, t_1 = C)$. See also Figure~\ref{fig:top-bottom}.
%\add{Fig. 9, 10, 11 and 12 should have as much consistency between them
%as possible.}

\begin{figure}
	\centering
	\includegraphics[width=.6\textwidth]{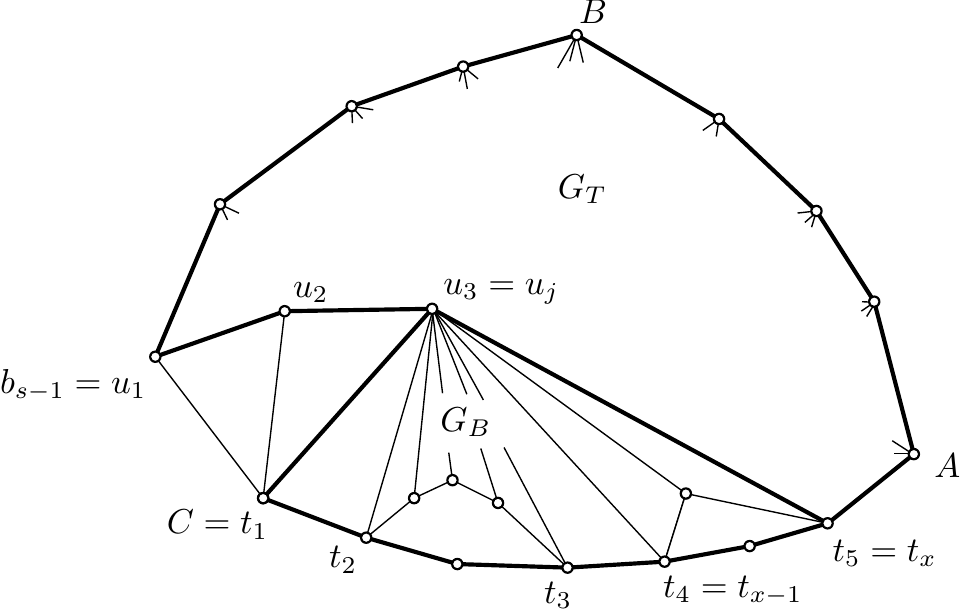}
	\caption{Decomposition into the top graph and bottom graph.}
        \label{fig:top-bottom}
\end{figure}

\begin{observation}
        The top graph $G_T$ is a W-triangulation that satisfies the chord condition with respect to corners
	$A':= A, B':= B$ and $C':=u_1$. 
\end{observation}
\begin{proof}
Since $u_2,\dots,u_j$ are interior vertices of $G$, the outer face of $G_T$
is a simple cycle, and so $G_T$ is a $W$-triangulation.

    Since $G$ satisfies the chord condition, graph $G_T$ does not have any 
chords with both ends on $P_{A'B'} = P_{AB}$ or $P_{B'C'} \subseteq P_{BC}$.
If there were any chords with both ends on $P_{C'A'}$, then by $C'=u_1$
the chord would either
connect two neighbours of $C$ (hence give a separating triangle of $G$),
or connect some $u_i$ for $i<j$ to $P_{CA}$ (contradicting minimality of $j$),
or connect $u_j$ to some other vertex on $P_{CA}$ (contradicting that
$t_x$ is the last terminal), or have both ends on $P_{CA}$ (contradicting the chord condition
for $G$).    Hence no such chord can exist either.
\end{proof}

So, we can apply induction on $G_T$ and obtain an
\int{(u_1,u_2)}-CZ-rep\-re\-sen\-ta\-tion of $G_T$.  Similarly as in previous cases
the plan is to combine this with a representation of the rest.
Define $G_Q:= G_B-u_j$; we call this graph the {\em chain graph}.%
\footnote{Comparing to the terminology of \cite{cit:ham-cycle}, our chain
graph is similar to the graph $G'$ that is defined by the $Q$-chain and
satisfies Property (A).}
Unfortunately $G_Q$ is not
necessarily 2-connected, and so we cannot apply induction to it directly,
but we can obtain a CZ-representation for it by splitting it into smaller subgraphs.

For $i=1,\dots,x-1$, let $G_i^+$ be the graph bounded by $(t_i,
\stackrel{{\normalsize P_{t_it_{i+1}}}}{\ldots}, t_{i+1},u_j,t_i)$, 
and let the {\em $i^{\mbox{\scriptsize{th}}}$ block}
be the graph $G_i := G_i^+-u_j$.  See also Figure~\ref{fig:chain_blocks}.

\begin{figure}
	\centering
	\includegraphics[width=.45\textwidth]{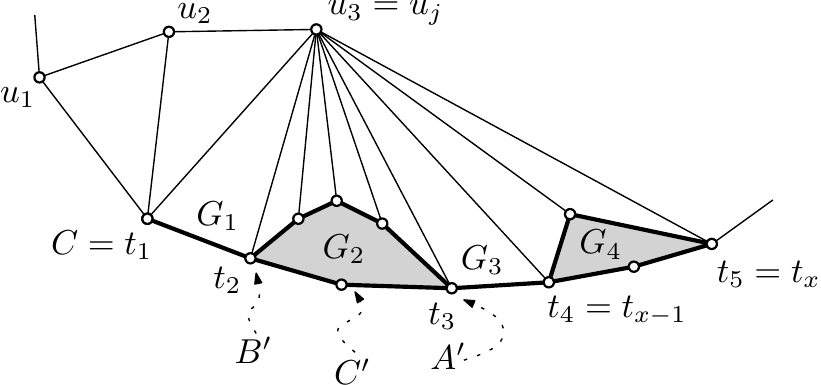}
	\caption{The chain graph. Blocks $G_1$ and $G_3$ are isolated
	edges, graphs $G_2$ and $G_4$ are W-triangulations.  We illustrate
	the chosen corners for $G_2$.}
	\label{fig:chain_blocks}
\end{figure}

\begin{observation}
\label{obs:chain-block}
The $i^{\mbox{\scriptsize{th}}}$ block $G_i$ is either a single edge, or a W-triangulation
that satisfies the chord condition with respect to corners
$A' := t_{i+1}, B' := t_{i}$, and $C'$ the successor of $t_i$ on $P_{t_i t_{i+1}}$.
\end{observation}
\begin{proof}
Assume $G_i$ is not a single edge.  First note that no vertex can appear twice
on the outer face boundary of $G_i$, otherwise (since $G$ was a triangulated disk)
there would be a double edge to $u_j$, or another terminal on $P_{t_i t_{i+1}}$.
So $G_i$ is bounded by a simple cycle and hence a W-triangulation.

Before we can argue the chord condition, we must see that the corners are
distinct.  Clearly $B'=t_i\neq t_{i+1}=A'$ by the definition of terminals.
Also $C'\neq B'=t_i$ by definition of $C'$ as a neighbour of $t_i$.  
Finally, $C'\neq A'=t_{i+1}$, for otherwise $(t_i,t_{i+1})$ would be an
edge and $\{u_j,t_i,t_{i+1}\}$ hence a separating triangle of $G$ (since
$G_i$ is not a single edge).  Now we can verify the chord condition:
\begin{itemize}
\item $G_i$ has no chord on $P_{A' B'}$, since all vertices on $P_{A' B'}$ are 
	neighbours of $u_j$ and $G$ has no separating triangle.  
\item $G_i$ has no chord on $P_{B' C'}$ or $P_{C' A'}$, since both of
	these are sub-paths of $P_{CA}$ and $G$ satisfies the chord
	condition.
\end{itemize}
\end{proof}

Set $F_1:=F$ if $G_1$ is not a single edge and $F_1=\emptyset$ otherwise.
Set $F_i=\emptyset$ for $1<i<x$.  By induction, any $G_i$
that is not an edge has an \int{F_i}-CZ-representation.
If $G_i$ is a single edge $(t_i,t_{i+1})$, then we can represent it with two
vertical segments for $t_i$ and $t_{i+1}$.  We now merge these representations
of $G_1,\dots,G_{x-1}$ as illustrated in Figure~\ref{fig:chain-representation},
by merging for $i=2,\dots,x-1$ the two vertical segments $\bb{t_i}$.
The result satisfies all conditions for an 
\int{F_1}-CZ-representation 
of $G_Q$ (for corners $B':=C$ and $A':=t_x$).
%
%except that the graph was not a triangulated disk;  slightly abusing notation
%we nevertheless call this an \int{F_1}-CZ-representation of $G_Q$. 

\begin{figure}
    \centering
    \raisebox{-0.5\height}{\includegraphics[width=.4\textwidth]{chain}}
\raisebox{-0.5\height}{~~~$\rightarrow$~~~}
    \raisebox{-0.5\height}{\includegraphics[width=.35\textwidth]{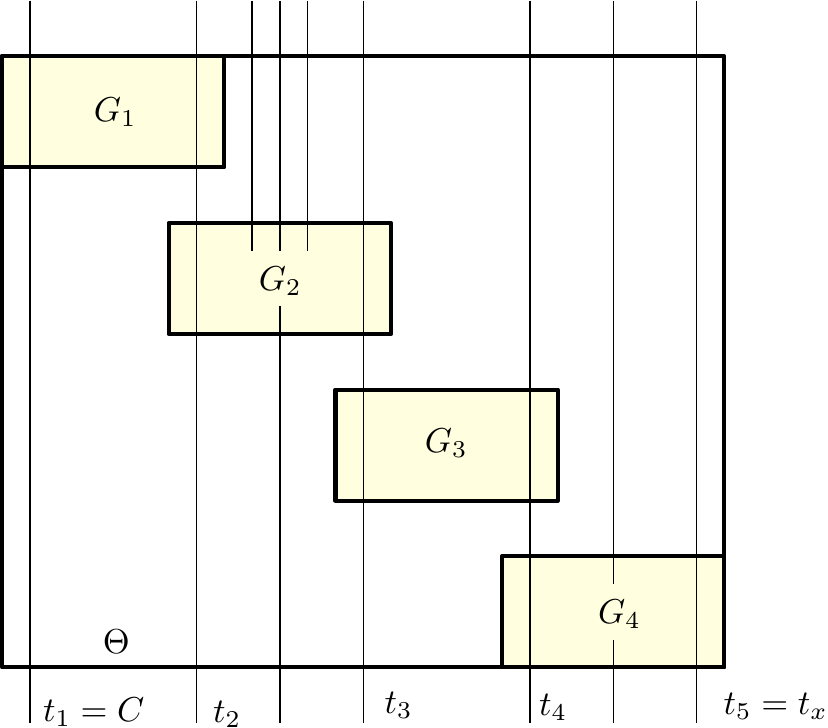}}
    \caption{Merging CZ-representations of the blocks to get a CZ-representation
of $G_Q$.}
    \label{fig:chain-representation}
\end{figure}

Now we merge this \int{F_1}-CZ-representation of $G_Q$ with the 
\int{(u_1,u_2)}-CZ-representation of $G_T$ as illustrated in
Figure~\ref{fig:case5extended}.
The neighbours of $G_Q$ in $G_T$ are vertices $u_1, u_2, \ldots, u_j$. 
The CZ-representation of $G_Q$ can be horizontally deformed and aligned so that 
it is below $G_T$ and the order of vertical segments is as in 
Figure~\ref{fig:case3-complete},
i.e., from left to right it is $\bb{u_1}, \bb{C}, \bb{u_2}, \bb{u_3}, \ldots, \bb{u_j}$, 
\{\bb{v} for $v\neq C$ in $G_Q$'s boundary path
$P_{t_xC}$ in reverse  order\},
\{\bb{c} for $c\neq t_x$ in $G$'s boundary path
$P_{t_xA}$ in order\}.

To create the required intersections, we first
extend $\bb{C}$ upward and (still below $G_T$) rightward until it crosses $\bb{u_2}, \ldots, \bb{u_j}$.
Note that as a result, $\bb{C}$ receives its first bend.

The CZ-representation of $G_T$ includes an intersection for the special
edge $(u_1,u_2)$, but does not include intersections for edges 
$({u_i},{u_{i+1}}), 2 \leq i < j$.
Since these are internal to $G$, such intersections need to be created. Do this for $i = 2,\ldots,j-1, j \geq 3$
by extending $\bb{u_i}$ vertically below its crossing with $\bb{C}$ and $\bb{u_{i-1}}$ and then horizontally
until it crosses $\bb{u_{i+1}}$ (which is adjacent due to the properties of a CZ-representation).
We assume for this that $\bb{u_{i+1}}$ has been extended vertically suitably. 
This increases the number of bends of $\bb{u_i}$ to at most 2, which is allowed
since $u_i$ is an internal vertex of $G$.

Curve $\bb{u_j}$ already extends below its crossing with $\bb{C}$ and
$\bb{u_{j-1}}$. Now, extend
it rightward so that it crosses all vertical segments of up to (and including) $\bb{t_x}$.
All these segments belong to vertices in $G_Q$'s boundary path $P_{t_xC }$,
which are indeed neighbours  of $u_j$ since $G$ is a triangulated disk.
Afterwards, $u_j$ has up to two bends, which is allowed since
it is not on the outer face of $G$.

It remains to create intersections for some of the edges in $G_Q$'s boundary
path $P_{t_x C}$.
We can create these intersections similarly as for $\bb{u_2}, \ldots, \bb{u_j}$
by extending each curve upward and rightward until it hits the next one.
This adds one bend in each curve, which is acceptable since the curve 
remains on the outer face only if it was a terminal before, hence had no bends 
previously.   Note that an edge in $P_{t_x C}$
possible does not need an intersection (see e.g.~edge $(t_3,t_4)$ in 
Figure~\ref{fig:case3-complete}), namely, if the $i^{\mbox{\scriptsize{th}}}$ block consists of
a single edge $(t_i,t_{i+1})$ and this edge is not in $F$.  In this case,
simply stop $\bb{t_i}$ before it intersects $\bb{t_{i+1}}$.

With this, all interior edges of $G$ receive exactly one intersection of
their corresponding curves.  If $F$ is non-empty, then the special edge
$(C,c_2)$ received its intersection either via the \int{F_1}-CZ-representation
of $G_1$ (if $G_1$ is a triangulated disk), or $(C,c_2)=(t_1,t_2)$ and 
the intersection was created when handling $P_{t_x C}$.
 
This ends the description of constructing an 
\int{F}-CZ-representation in Case~\ref{case:chordless},
and hence proves Lemma~\ref{lem:representation} and
Theorem~\ref{thm:main-claim}.

\begin{figure}
	\centering
	\includegraphics[width=.7\textwidth]{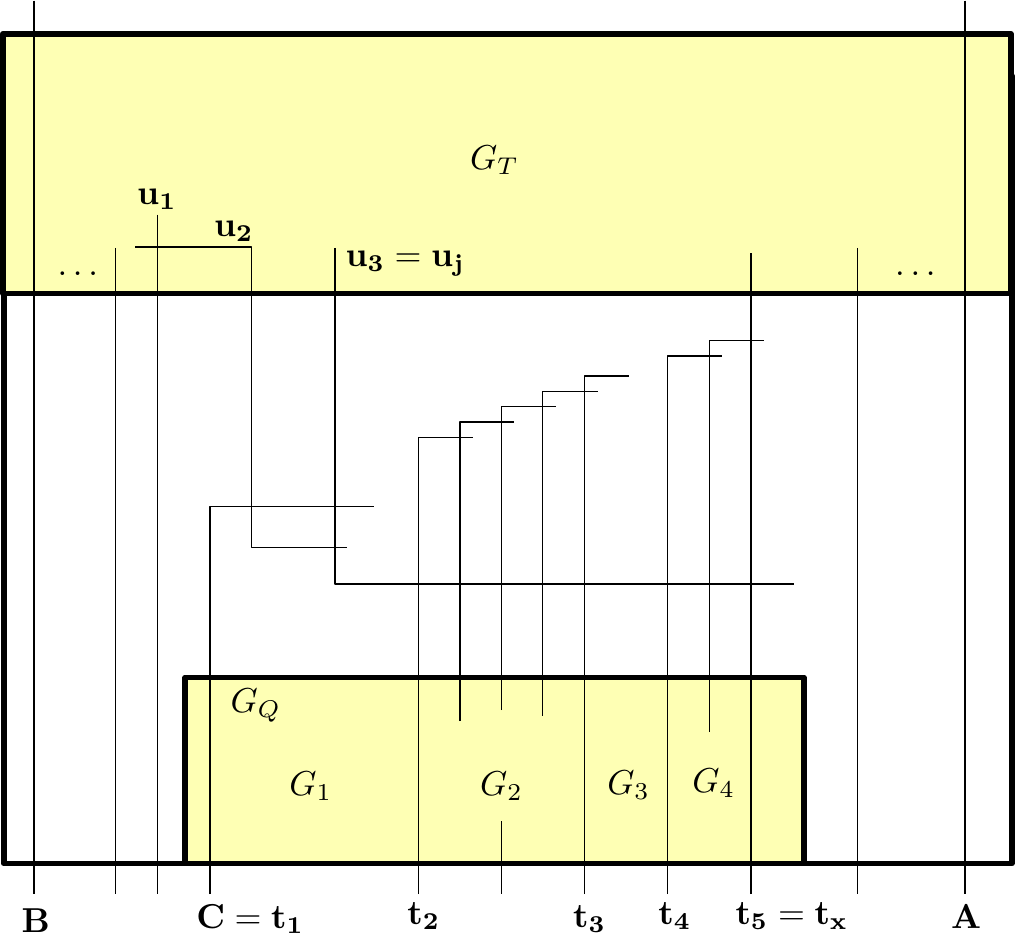}
	\caption{Obtaining an \int{F}-CZ-representation in Case~\ref{case:chordless}.}
	\label{fig:case3-complete}
\end{figure}
%\add{Fig 12: add $u_j$ (perhaps also $u_{j-1}$) and delete $c_i$ and $c_{i+1}$.}

\begin{figure}
	\centering
        \raisebox{-0.5\height}{\includegraphics[height=.45\textwidth]{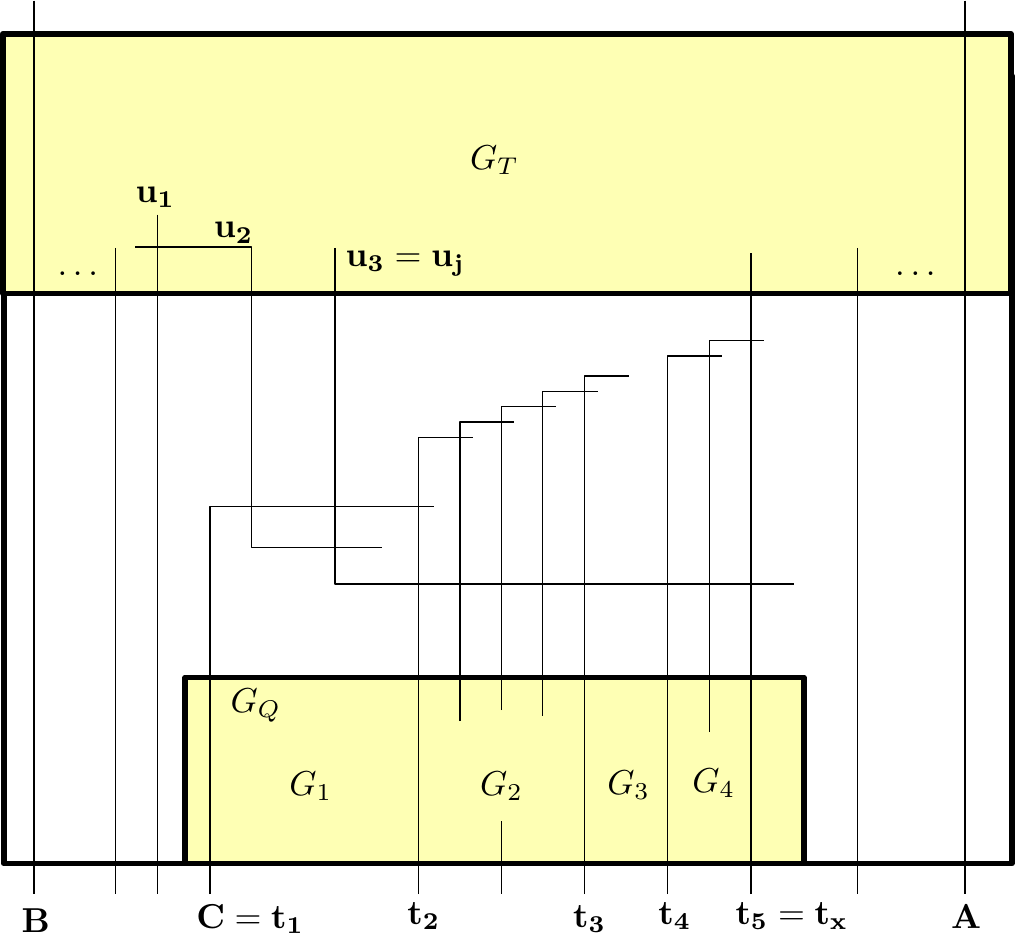}}
~~~~~~~~~~~
        \raisebox{-0.5\height}{\includegraphics[height=.46\textwidth,trim=260 0 0 0,clip]{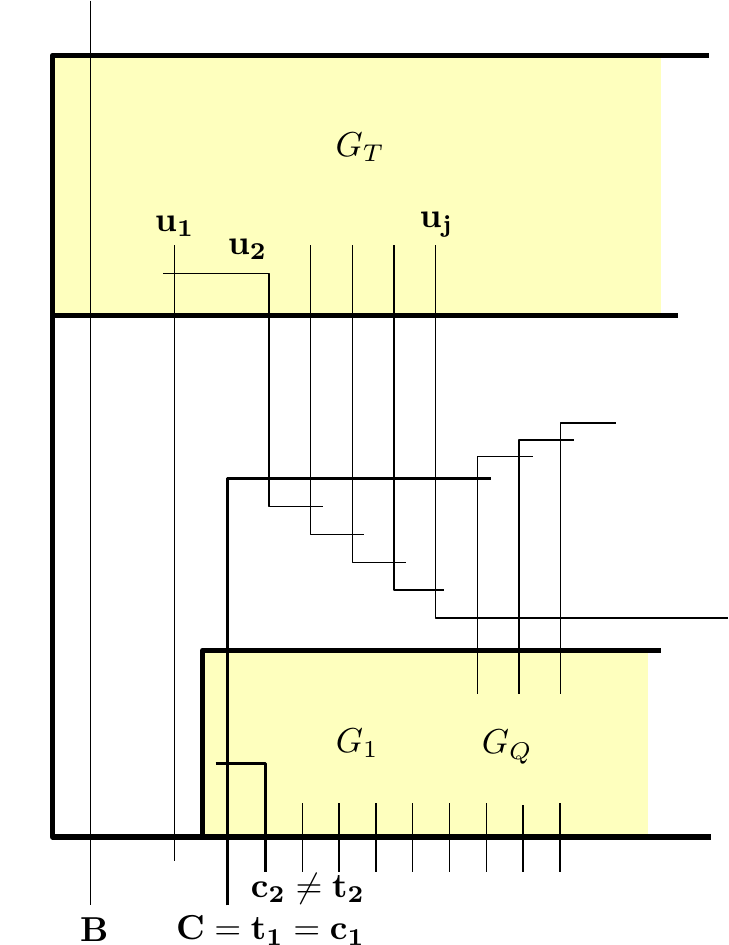}}
	\caption{If $F=(C,c_2)$, then the intersection of $\bb{C}$ and $\bb{c_2}$ is added during the merge if $G_1$ is a single edge (left), or occurs within
the CZ-representation of $G_1$ otherwise (right).}
        \label{fig:case5extended}
\end{figure}

\section{Conclusions and Outlook}
\label{sec:conclusions}
\label{sec:outlook}

We showed that every 4-connected planar graph has a 1-string CZ-representation. A natural 
question is to extend this result to all planar graphs. We believe that this is possible,
but currently cannot achieve fewer than $k=3$ bends.

\iffalse
\begin{observation}[Chalopin, Gon\c{c}alves, Ochem~{\cite[Lemma~1]{cit:chalopin-string}}]
\label{obs:triangulation-only}
    Every planar graph $G$ is an induced subgraph of a~$3$-connected 
    planar triangulation. 
\end{observation}
\fi
    
By again stellating the graph (possibly repeatedly if it was not 3-connected),
it suffices to show that every $3$-connected
planar triangulation has a $1$-string $B_k$-VPG representation. This statement can be proved by induction
on the number of separating triangles by a technique used in~\cite{cit:chalopin-string}
(and re-discovered in~\cite{cit:mfcs}). With every triangular face,
create a ``face region'' (called ``private region'' in \cite{cit:mfcs})
that intersects the curves of vertices of the face
in a predefined way and does not intersect anything else. This is easy for 
W-triangulations by inspecting the constructions in Cases~\ref{case:special}--\ref{case:chordless}.

In the inductive step, find the smallest separating triangle $T$ in $\Delta$.
By induction, the graph $G_1$ obtained by removing the inside of
$T$ has a 1-string $B_k$-VPG representation. 
The graph $G_2$ strictly inside $T$ is either a single vertex or (as one
can show) it is a W-triangulation that satisfies the chord condition for
some suitably chosen corners, and hence has an int-CZ-representation.
Place it inside the face region for $T$,
create the intersections needed for edges on the outer face of $G_2$ and
edges between $G_1$ and $G_2$, and
identify face regions for newly created faces.

If one aims for a 1-string $B_3$-VPG representation, two bend per
can be added to each curve of outer face vertices of $G_2$, and hence 
such a merge is easy (details are omitted).

If one aims for a 1-string $B_2$-VPG representation of planar
graphs, only one bend can be added to each such curve, 
which seems impossible with our current 
geometric restrictions of \int{F}-CZ-representation. 
However, this might be feasible
if we allow the outer face vertices 
to use rays in three directions. This is our ongoing research.

\medskip
As for other future work,
the CZ-representation constructed in this paper uses curves of four possible shapes.
Is it possible to use fewer shapes or to restrict them further?  
Felsner et al.~\cite{cit:mfcs} asked the question whether every planar
graph is the intersection graph of only two shapes, namely $\{L,\Gamma\}$.
(This would also provide a different proof of Scheinerman's conjecture.)
Somewhat inbetween: is every planar graph the intersection graph of
$xy$-monotone orthogonal curves, preferably in the 1-string model?

\bibliography{arxive}{}
\bibliographystyle{plain}

\end{document}